\documentclass[11pt]{article} 

\usepackage[utf8]{inputenc} 

\usepackage{geometry} 
\usepackage{graphicx} 

\usepackage{amsmath}
\usepackage{amsfonts}
\usepackage{amssymb}
\usepackage{mathrsfs}
\usepackage{amsthm}
\usepackage[symbol]{footmisc}

\newcommand{\C}{\mathbb{C}}

\newcommand{\N}{\mathbb{N}}

\newcommand{\leftset}{\left\{\left.}
\newcommand{\midsetl}{\right\vert\left.}

\newcommand{\rightset}{\right.\right\}}
\newcommand{\style}{\mathcal}
\newcommand{\tr}{\text{Tr}}
\newcommand{\Tr}{\tr}
\newcommand{\Md}{\mathcal{M}_d}
\newcommand{\UCP}{\style{E}}

\newcommand{\Me}{\mathcal{M}_{\mathcal{E}}}
\newcommand{\M}{\mathcal{M}}
\newcommand{\Ep}{\mathcal{E}}

\newtheorem{definition}{Definition}[section]

\newtheorem{lemma}{Lemma}[section]
\newtheorem{cor}{Corollary}[section]
\newtheorem{theorem}{Theorem}[section]
\newtheorem{prop}{Proposition}[section]
\newtheorem{ex}{Example}[section]
\newtheorem{remark}{Remark}[section]
\newtheorem{conjecture}{Conjecture}[section]

\newcommand{\footremember}[2]{%
    \footnote{#2}
    \newcounter{#1}
    \setcounter{#1}{\value{footnote}}%
}
\newcommand{\footrecall}[1]{%
    \footnotemark[\value{#1}]%
}

\title{Eventually Entanglement Breaking Maps}
\author{Mizanur Rahaman\footremember{pmath}{Department of Pure Mathematics, University of Waterloo}\footremember{iqc}{Institute for Quantum Computing, University of Waterloo}, Sam Jaques\footrecall{iqc} \footremember{co}{Department of Combinatorics and Optimization, University of Waterloo}, and Vern I. Paulsen\footrecall{pmath} \footrecall{iqc}}

\begin{document}
\maketitle
\begin{abstract}
We analyze linear maps on matrix algebras that become entanglement breaking after composing a finite or infinite number of times with themselves. This means that the Choi matrix of the iterated linear map becomes separable in the tensor product space. 
If a linear map becomes entanglement breaking after finitely many iterations, we say the map has a finite index of separability. In particular we show that every unital PPT-channel has a finite index of separability and that the class of unital channels that have finite index of separability is a dense subset of the unital channels.
 We construct concrete examples of maps which are not PPT but have finite index of separability. We prove that there is a large class of unital channels that are asymptotically entanglement breaking. This analysis is motivated by the PPT-squared conjecture made by M. Christandl that says every PPT channel, when composed with itself, becomes entanglement breaking.
\end{abstract}
\section{Introduction and Motivation}
Detecting separability is one of the key aspects of the theory of quantum information. Computationally it is a hard problem to decide whether a state (a positive semi-definite matrix of trace 1) is separable or not. The quantum channel associated to a state is entanglement breaking if the (Choi)state is separable.

  An important class of linear maps in this context is the class of PPT maps. The Choi matrices of such maps are positive and have a positive partial transpose. These states turn out to be useful in Quantum Key Distribution (QKD) and dense coding protocols (see \cite{B-C-W}).  
It was observed that all known PPT maps, when composed twice with themselves, become entanglement breaking. Hence a conjecture was put forward by M. Christandl in \cite{Banff}
\begin{conjecture}[{\rm{PPT-squared Conjecture}}]\label{conj-PPT}
The composition of two PPT maps is always entanglement breaking.
\end{conjecture}
Our work is closely related to that of Lami and Giovannetti in \cite{Lami-G}, who studied quantum channels that are ``entanglement saving", i.e., where no power of the map becomes entanglement breaking. They studied such maps in great detail and also provided conditions on when a channel is ``asymptotically entanglement saving". Although they did not study PPT maps in particular, some of our results concerning PPT maps can be deduced from their work. Our proofs use multiplicative domains and are more direct for our particular results.
 Very recently in \cite{K-M-P}, the authors analyzed the composition of PPT maps and proved an asymptotic version of the above conjecture. They showed that for every PPT map $\UCP$, the sequence of distances $d(\UCP^n, \rm{EB})$ goes to zero as $n\to \infty$, where $\rm{EB}$ is the set of entanglement breaking maps. This finding verifies the conjecture asymptotically, that is, $\UCP$ is asymptotically entanglement breaking.
 
 In this paper we look at PPT maps as a particular class of maps which become entanglement breaking after finitely many iterations. If a map $\UCP$ is such that $\UCP^n$ is entanglement breaking for some $n\in \mathbb{N}$, then we say that $\UCP$ has finite index of separability.
 The spectral and multiplicative properties of unital channels are the key ingredients of this investigation. In Section \ref{sec:Structure of PPT Maps}, we provide a structure theorem (Theorem \ref{Thm-mult-abel-PPT}) for a unital PPT map with respect to its action on projections that are mapped to another projection. This reveals that its action can be restricted to independent blocks on the diagonal in a suitable basis.
 In Section \ref{sec:Composition of PPT maps and Finite Index of Separability}, we prove that every unital PPT channel becomes entanglement breaking after finite iterations (see Theorem \ref{thm-PPT-finite-index}). An illustrative example (Example \ref{ex1}) is put forward in Section \ref{sec:Beyond PPT maps} to show that there are maps which are not PPT yet they become eventually entanglement breaking after finite iterations. It is also shown in Theorem \ref{thm-denseness} that the class of such maps which have finite index of separability is dense in the class of unital channels. In Section \ref{sec:Asymptotically Entanglement Breaking}, a large class of channels are shown to be eventually entanglement breaking in asymptotic sense (see Theorem \ref{asympt-entgl}).
\section{Background}
We begin with defining a quantum channel which is a completely positive and trace preserving linear map $\UCP:\Md\rightarrow \Md$. Note that every linear map $\UCP$ on $\Md$
defines a unique element in $\Md\otimes\Md$ which is known as the Choi-matrix of $\UCP$ and it is defined as $C_\UCP=\displaystyle\sum_{i,j}^{d}E_{ij}\otimes\UCP(
E_{ij})$, where $E_{ij}$ are the matrix units in $\Md$. It holds that $\UCP$ is completely positive if and only if $C_{\UCP}$ is positive semi-definite in $\Md\otimes\Md$. For general introduction of completely positive maps and quantum channels we refer to the monographs \cite{paulsen}, \cite{Watrous-book}.

One of the main tools we use to investigate iterative properties of a channel is the multiplicative nature of the channel.
We note some definitions in this regard.
\begin{definition}
For a linear map $\UCP:\Md\rightarrow\Md$, the following set is known as the multiplicative domain of $\UCP$
\[\Me=\{a \in \Md: \UCP(xa)=\UCP(x)\UCP(a), \UCP(ax)=\UCP(a)\UCP(x), \ \forall  x\in \Md\}\]
\end{definition}
\begin{theorem}[\rm{Choi, \cite{choi1}}]
	For a unital completely positive map  $\UCP:\Md\rightarrow \Md$, the multiplicative domain of $\UCP$ is a C$^*$-subalgebra of $\Md$ and equals the following set 
	\[\Me=\{x\in \Md :   \UCP(x^*x)=\UCP(x^*)\UCP(x), \ \UCP(xx^*)=\UCP(x)\UCP(x^*)\}.\]
\end{theorem}

 Note that if the maps are assumed to be unital, trace preserving and completely positive, then the fixed point set of $\UCP$ is defined by
 \[\rm{Fix}_{\UCP}=\{a\in \Md: \UCP(a)=a\}.\]
This is an algebra and moreover we have $\rm{Fix}_{\UCP}\subseteq\Me$. 
For any unital channel $\UCP$ and any $k\in\N$, $\mathcal{M}_{\UCP^{k+1}}\subseteq \mathcal{M}_{\UCP^k}$ (See \cite{miza}), and hence there is some $N\in\N$ such that for any $n\geq N$, $\mathcal{M}_{\UCP^n}=\mathcal{M}_{\UCP^N}$. Following \cite{miza}, we denote this algebra $\mathcal{M}_{\UCP^\infty}$ and refer to it as the \textit{stabilized multiplicative domain} of $\UCP$.
\begin{definition}[\cite{miza}]\label{def:mult-index}
	The \textit{multiplicative index} of a unital quantum channel $\UCP$ is the minimum $n\in\N$ such that $\mathcal{M}_{\UCP^n}=\mathcal{M}_{\UCP^\infty}$. In other words it is the length of the following decreasing chain of subalgebras
\[\Me\supseteq \mathcal{M}_{{\UCP}^2}
\supseteq\cdots \supseteq\mathcal{M}_{{\UCP}^{\infty}}.\]
We often denote this number as $\kappa(\UCP)$.
\end{definition}
Another useful result is Lemma 2.2 from \cite{miza}:
\begin{lemma}[\rm{see \ \cite{miza}}]\label{lem:md_composition}
	If $\UCP_1,\UCP_2$ are two unital quantum channels, then
	\[\mathcal{M}_{\UCP_2\circ\UCP_1}=\{x\in \mathcal{M}_{\UCP_1}|\UCP_1(x)\in \mathcal{M}_{\UCP_2}\}.\]
\end{lemma}
 Next, we note down the definition of separability of a positive matrix.
 \begin{definition}A positive semi-definite matrix 
 $P\in \Md\otimes\mathcal{M}_{n}$ is called separable if there exists positive semi-definite matrices $R_1, \cdots, R_k\in \Md$ and $Q_1,\cdots, R_k\in \mathcal{M}_n$ such that $P=\displaystyle\sum_{i=1}^{k}R_i\otimes Q_{i}$.
 \end{definition}
 Separability of states has a close connection to a specific set of channels which we define below.
 \begin{definition}
 A linear map $\Phi:\Md\rightarrow\mathcal{M}_n$ is called entanglement breaking if the Choi-matrix of $\Phi$, 
 $C_{\Phi}$, is separable in $\Md\otimes\mathcal{M}_n$.
 \end{definition}
 The set of separable states is a compact convex set and hence so is the set of entanglement breaking maps.
There are many different equivalent criteria for a channel to be entanglement breaking. We note down the following facts about entanglement breaking maps.
 \begin{theorem}[\cite{entng-brkng}]\label{Thm-Horodecki}
For a linear map  $\Phi:\Md\rightarrow\mathcal{M}_n$, the following statements are equivalent:
\begin{enumerate}
\item $\Phi$ is entanglement breaking,
\item $\Phi$ has the form $\Phi(x)=\displaystyle \sum_{j=1}^{k}Tr(xR_k)Q_k$, where the $R_k$ are positive matrices with trace 1 in $\Md$ and $Q_1,\cdots Q_k$ are all positive semi-definite matrices in $\mathcal{M}_n$.
\item For any completely positive map $\Psi$, the maps $\Phi\circ\Psi$ and $\Psi\circ\Phi$ are entanglement breaking, whenever the composition is meaningful.
\end{enumerate}
 \end{theorem}

 \begin{definition}
 A linear map $\UCP:\Md\rightarrow\Md$ is said to have finite index of separability if there exists a $n\in \mathbb{N}$ such that $\UCP^n$ is entanglement breaking.
 \end{definition}
 \begin{remark}
  If a linear map $\UCP_1$ is entanglement breaking, then for any other completely positive map  $\UCP_2$, $\UCP_2\circ\UCP_1$ is entanglement breaking by Theorem \ref{Thm-Horodecki}. Thus, having a finite index of separability means that, after enough repeated applications, a linear map becomes entanglement breaking and stays entanglement breaking. Notice here that this index of separability for channels has been studied before in \cite{Lami1}. We further analyze this property keeping the PPT-squared conjecture in mind. 
\end{remark} 	

\section{Structure of PPT Maps}\label{sec:Structure of PPT Maps}
 In this section we analyze some essential features of PPT maps which eventually help us proving that all PPT maps have finite index of separability.
 \begin{definition}
 A linear map $\Phi:\Md\rightarrow\mathcal{M}_n$ is called PPT if it is completely positive and co-completely positive, that is, $\Phi\circ t$ is completely positive where $``t"$ is the transpose map $x\mapsto x^{t}$.
 \end{definition}
 Note that in the literature of quantum information theory, the name PPT appears with different meanings \cite{PPT-mapping},\cite{PPT-binding}. We will consider the above definition as PPT maps. It turns out that the Choi matrix of a PPT map is `positive partial transpose' in the tensor product $\Md\otimes\mathcal{M}_n$ (see Theorem 7.2.2 in \cite{stormer}).
 
 To decide whether a channel is entanglement breaking or not amounts to deciding whether the Choi matrix is separable, which is computationally 
 quite a difficult task. Note that since every separable state is positive partial transpose 
 \cite{entng-brkng}, it is clear that the set of entanglement breaking channels is a subset of the PPT channels. However, there are PPT channels that are not entanglement breaking. One more criterion for separability in terms of the multiplicative domain is recorded below.
\begin{theorem}[\rm{St\o rmer}, \cite{stormer2008}]
\label{stormer}
Let $\UCP:\Md\rightarrow\Md$ be a unital channel. Then if $\UCP$ is entanglement breaking, then the multiplicative domain $\Me$ is an abelian C$^*$-
algebra. The converse is also true provided $\UCP$ is a conditional expectation, that is, $\UCP^2=\UCP$.
\end{theorem}

Next we extend the results of St\o rmer's mentioned in Theorem \ref{stormer}. Before that we prove a lemma that will be useful in the subsequent discussion. This result was given in \cite{J-K-P-P} but for the sake of completeness we outline a proof. See also Corollary 3 in \cite{stormer2008}.
\begin{lemma}\label{abelian-range}
Let $\UCP:\Md\rightarrow\Md$ be unital completely positive map such that the $\rm{Range}(\UCP)$ is contained in an abelian C$^*-algebra$, then $\UCP$ is entanglement breaking.
\end{lemma}
\begin{proof}
Let the (finite dimensional) abelian C$^*$-algebra be $l^{\infty}_{k}$ for some $k\in \mathbb{N}$. Hence we can regard $\UCP(x)=(\UCP_1(x),\cdots, \UCP_k(x))$ for every $x\in \Md$. Now for each $1\leq i\leq k$, define the co-ordinate projections $p_{i}:l^{\infty}_k\rightarrow\mathbb{C}$, that is, $p_{i}(x_1,\cdots,x_k)=x_i$. These are positive linear functionals on $l^{\infty}_k$. Hence composing with $\UCP$, we get $p_i\circ\UCP:\Md\rightarrow\mathbb{C}$ which is a unital positive functional on $\Md$. Since every such functional is given by $x\rightarrow \Tr(xa_i)$ for some density matrix $a_i$, we have $\UCP_i(x)=Tr(xa_i)$. Now choose $k$ many positive elements $\{b_1,\cdots b_k\}$ in $\Md$ to associate each coordinate element of $l^{\infty}_k$ to $\Md$ and eventually we get 
\[\UCP(x)=\displaystyle\sum_{j=1}^{k}\Tr(xa_j)b_j, \ \forall x\in \Md,\]
which is entanglement breaking.
\end{proof}
Now we state and prove the main theorem of this section.
\begin{theorem}\label{Thm-mult-abel-PPT}
Let $\UCP:\Md\rightarrow\Md$ a unital PPT map. Let $p\in \Me$ be projection and let $q=1-p$. Then for every $x\in \Md$ we have 
\[\UCP(x)=\UCP(pxp)+\UCP(qxq)=\UCP(p)\UCP(x)\UCP(p)
+\UCP(q)\UCP(x)\UCP(q).\]
It follows that $\UCP(\Me)$ is an abelian C$^*$-algebra contained in the center of the C$^*$-algebra generated by the $\rm{Range}(\UCP)$. If moreover, $\UCP$ is faithful, then $\Me$ is abelian.  
\end{theorem}
\begin{proof}
If $\Me$ contains no projection then our conclusion follows trivially. So assume there exists a projection $p\in \Me$. So $\Ep(p)$ is a projection (that is $\Ep(p)^2=\Ep(p)$) and for all $x\in \Md$ we have $\Ep(xp)=\Ep(x)\Ep(p)$.
For any $x\in \Md$, define
\[X=\left[\begin{array}{cc}
p^t & 0 \\
x & 0 \\
\end{array}\right], \ 
\text{then} \ X^*=\left[\begin{array}{cc}
\bar{p} & x^* \\
0 & 0 \\
\end{array}\right].\]  

Hence $XX^*=\left[\begin{array}{cc}
p^t\bar{p} & p^{t}x^* \\
x\bar{p} & xx^* \\
\end{array}\right].$\\
By  2-positivity  of $\UCP\circ t$ we have by Schwarz inequality
\[\left[\begin{array}{cc}
\UCP(p) & 0 \\
\UCP(x^t) & 0 \\
\end{array}\right]\left[\begin{array}{cc}
\UCP(p) & \UCP(\bar{x}) \\
0 & 0 \\
\end{array}\right]\leq \left[\begin{array}{cc}
\UCP\circ t(p^{t}\bar{p}) & \UCP\circ t(p^{t}{x}^*) \\
\UCP\circ t(x\bar{p}) & \UCP\circ t({x}x^*) \\
\end{array}\right]=\left[\begin{array}{cc}
\UCP(pp) & \UCP(\bar{x}p) \\
\UCP(px^t) & \UCP(\bar{x}{x^t}) \\
\end{array}\right]\]
This implies 
\[\left[\begin{array}{cc}
\UCP(p) & \UCP(p)\UCP(\bar{x}) \\
\UCP(x^t)\UCP(p) & \UCP(x^t)\UCP(\bar{x}) \\
\end{array}\right]\leq \left[\begin{array}{cc}
\UCP(p) & \UCP(\bar{x}p) \\
\UCP(px^t) & \UCP(\bar{x}x^t) \\
\end{array}\right].\]
This yields
\[\left[\begin{array}{cc}
0 & \UCP(\bar{x}p)-\UCP(p)\UCP(\bar{x}) \\
\UCP(px^t)-\UCP(x^t)\UCP(p) & \UCP(\bar{x}x^t)-\UCP(x^t)\UCP(\bar{x}) \\
\end{array}\right]\geq 0.\]

If the $(1,1)$ entry of a positive is zero, then the $(1,2)$ entry will be zero. Hence
$\UCP(\bar{x}p)=\UCP(p)\UCP(\bar{x})$. Since $x$ was arbitrary, replacing $x$ by $\bar{x}$, we get 
\[\UCP(xp)=\UCP(p)\UCP(x).\]
But we already have $\UCP(xp)=\UCP(x)\UCP(p)$ from the multiplicative domain property. Hence we obtain for every $x\in \Md$ and any projection $p\in \Me$, 
\[\UCP(x)\UCP(p)=\UCP(p)\UCP(x).\]
So it follows that for every $x \in \Md$, we have two projections $p,q=(1-p)\in \Me$ ($p\perp q$) such that 
\begin{align*}
\UCP(x)&=\UCP((p+q)x(p+q))\\&=\UCP(pxp)+\UCP(qxp)+\UCP(pxq)
+\UCP(qxp)\\
&=\UCP(p)\UCP(x)\UCP(p)+\UCP(q)\UCP(x)\UCP(q).
\end{align*}
The other terms vanishes as $\UCP(p)\perp\UCP(q)$ and the above commutation property. It also shows that $\UCP(\Me)$ is abelian and it is contained in the center of the algebra generated by $\rm{Range}(\UCP)$.
Now if $\UCP$ is faithful, restricting $\UCP$ on $\Me$ we get an isomorphism between $\Me$ and $\UCP(\Me)$. As $\UCP(\Me)$ is abelian, it follows that $\Me$ is abelian.
\end{proof}
We get an immediate useful corollary to the above theorem.
\begin{cor}\label{lemma:coCP_md_struct}
	Let $\UCP:\Md\rightarrow \Md$ be a unital PPT map. If $p,q$ are two orthogonal projections in the multiplicative domain $\Me$ and $x\in \Md$ such that $pxq=x$, then $\UCP(x)=0$.
\end{cor}
\begin{proof}
If $x=pxq$, then from the previous theorem we get $\UCP(x)=\UCP(p)\UCP(x)\UCP(q)=\UCP(x)\UCP(p)\UCP(q)=0$.
\end{proof}
This result shows that for every PPT map, there is some basis where its action can be restricted to independent blocks on the diagonal.\par

Next we extend the results of St\o rmar given in Theorem \ref{stormer} for PPT maps.
\begin{theorem}\label{PPT-multi-dom}
If $\UCP:\Md\rightarrow\Md$ is a unital PPT channel, then $\Me$ is abelian and the converse is also true if $\UCP$ is a conditional expectation.
\end{theorem}
\begin{proof}
The first part is essentially Theorem \ref{Thm-mult-abel-PPT} as the trace preservation property implies faithfulness. 

Conversely, suppose $\UCP$ is a channel such that $\UCP^2=\UCP$ and the multiplicative domain $\Me$ is abelian. Then we have that $\rm{Range}(\UCP)$ is contained in $\Me$. Indeed, for any $a\in \rm{Range}(\UCP)$ we have 
\[\UCP(\UCP(aa^*)-aa^*)=0.\]
But $\UCP(aa^*)-aa^*\geq 0$ by the Schwarz inequality of $\UCP$. Now by the trace preservation of $\UCP$, from the above equation we get $\UCP(aa^*)=aa^*=\UCP(a)\UCP(a^*)$, which is the equality in the Schwarz inequality and hence $a\in \Me$.
As $\rm{Range}(\UCP)$ is contained in an abelian C$^*$-algebra, by Lemma \ref{abelian-range} we get $\UCP$ is entanglement breaking. Since an entanglement breaking map is automatically PPT, we have the result.

\end{proof} 
The composition of PPT maps will have a similar structure, with independent projections, but the projections may not ``line up'', and so we need to analyze the composition of PPT maps in  more detail. 

\section{Composition of PPT maps and Finite Index of Separability}\label{sec:Composition of PPT maps and Finite Index of Separability}
In this section we prove one of the main results of the paper, that is, every unital PPT map becomes entanglement breaking after a finite number of iterations.

\begin{theorem}\label{thm:PPT_Minf_decomp}
Let $\UCP:\Md\rightarrow\Md$ be a unital PPT channel and let $p_1,\cdots,p_k$ be the set of minimum mutually orthogonal projections $\M_{\UCP^\infty}$. Then there is some natural number $n$ such that
\[\UCP^n=\UCP_1\oplus\cdots\oplus \UCP_k\]
where $\UCP_i:p_i\Md p_i\rightarrow p_i\Md p_i$ is a unital PPT channel with trivial multiplicative domain. 
\end{theorem}
\begin{proof}
Let $\kappa$ be the multiplicative index of $\UCP$ defined as in \ref{def:mult-index}. Then the multiplicative domain of $\UCP^\kappa$ is just $\M_{\UCP^\infty}$. As $\M_{\UCP^\infty}\subseteq\Me$, by Theorem \ref{Thm-mult-abel-PPT}, this domain is an abelian subalgebra of $\Md$, and thus there is some set of minimum mutually orthogonal projections $p_1,\cdots, p_k$ that span $\M_{\UCP^\infty}$. We can apply Theorem \ref{Thm-mult-abel-PPT} repeatedly and get that the action of $\UCP^\kappa$ is
\begin{align}
\UCP^\kappa(x)=&\UCP^\kappa(p_1xp_1)+\cdots+\UCP^\kappa(p_kxp_k)\nonumber\\
=&\UCP^\kappa(p_1)\UCP^\kappa(p_1xp_1)\UCP^\kappa(p_1)+\cdots+\UCP^\kappa(p_k)\UCP^\kappa(p_kxp_k)\UCP^\kappa(p_k).\label{eq:PPTthm}
\end{align}
Since $p_i\in\M_{\UCP^\infty}$, then $\UCP^\kappa(p_i)$ is also a projection in $\M_{\UCP^\infty}$. If we take a projection of minimum rank, $p_i$, then $\UCP^\kappa(p_i)$ must also be a projection of minimum rank. Since these projections are all accounted for in $p_1,\cdots p_k$, then $\UCP^\kappa(p_i)=p_{\sigma(i)}$ for some permutation $\sigma:[k]\rightarrow[k]$. We can repeat this process for every projection of minimum rank, and then for the projections of second-smallest rank, until we construct a permutation $\sigma$ such that for all $i$, $\UCP^\kappa(p_i)=p_{\sigma(i)}$. Then, for any $x\in\Md$, we can use Equation \ref{eq:PPTthm}:
\begin{align*}
\UCP^{2\kappa}(x)=&\UCP^{\kappa}\left(p_{\sigma(1)}\UCP^\kappa(p_1xp_1)p_{\sigma(1)}+\cdots+p_{\sigma(k)}\UCP^\kappa(p_kxp_k)p_{\sigma(k)}\right)\\
=&\UCP^{\kappa}(p_{\sigma(1)})\UCP^{2\kappa}(p_1xp_1)\UCP^{\kappa}(p_{\sigma(1)})+\cdots+\UCP^{\kappa}(p_{\sigma(k)})\UCP^{2\kappa}(p_kxp_k)\UCP^{\kappa}(p_{\sigma(k)})\\
=&p_{\sigma^2(1)}\UCP^{2\kappa}(p_1xp_1)p_{\sigma^2(1)}+\cdots+p_{\sigma^2(k)}\UCP^{2\kappa}(p_kxp_k)p_{\sigma^2(k)}.
\end{align*}
Then $\sigma$ will have some finite order $m$, so repeating the calculations above gives:
\[\UCP^{m\kappa}(x)=p_1\UCP^{m\kappa}(p_1xp_1)p_1+\cdots+p_k\UCP^{m\kappa}(p_kxp_k)p_k.\]
Setting $n=m\kappa$ and $\UCP_i(x)=p_i\UCP^{m\kappa}(x)p_i$ for $x\in p_i\M_d p_i$, and applying Theorem \ref{Thm-mult-abel-PPT} to $\UCP^{m\kappa}$, which has multiplicative domain equal to $\M_{\UCP^\infty}$, gives the required result.
\end{proof}

At this juncture we note down some results that we will use to arrive to the main theorem of the section.
\begin{theorem}{\rm{(Gurvits-Barnum}, \cite{Gurvits})}\label{gurvits}
Let $\rho$ be a normalized density matrix in a bipartite system of total dimension $d=nm$ such that 
$\|\rho-\frac{1}{d}(1_n\otimes1_m) \|_{2}^{2}\leq \frac{1}{d(d+1)}$, then $\rho$ is separable, where $1_n, 1_m$ are the identity matrices in $\mathcal{M}_n$ and $\mathcal{M}_m$ respectively.
\end{theorem}
Note that $1_n\otimes1_m$ is the Choi matrix for the channel $\Omega:\mathcal{M}_n\rightarrow\mathcal{M}_m$ given by $\Omega(x)=\tr(x)\frac{1_m}{n}$, for all $x\in \mathcal{M}_m$. Using the Choi-Jamiolkowski identification with matrices in $\mathcal{M}_n\otimes\mathcal{M}_m$ and the linear maps from $\mathcal{M}_n$ to $\mathcal{M}_m$, the above theorem ensures that the map $\Omega$ has a neighborhood where each quantum channel in the neighborhood is entanglement breaking.

In what follows $\mathbb{T}$ represents the unit circle in the complex plane. Note that (see \cite{wolf}) for any quantum channel $\UCP$, all the eigenvalues lie in the closed unit disc of the complex plane. We define the peripheral spectrum ($\rm{Spec}_{\UCP}$) of $\UCP$ as follows
\[\rm{Spec}_{\UCP}=\{\lambda\in \mathbb{C} \ | \ (\lambda. \rm{id}-\UCP) \ \text{is \ not \ \  invertible \ on} \ \M_d \},\] 
where $``\rm{id}"$ is the identity operator on $\Md$. The set $\rm{Spec}_{\UCP}\cap \mathbb{T}$ is called the \emph{peripheral eigenvalues} and any $a\in \Md$ satisfying $\UCP(a)=\lambda a$, with $|\lambda|=1$, is called a \emph{peripheral eigenvector} corresponding $\lambda$. It is a consequence of the Kadison-Schwarz inequality that for a unital channel $\UCP$, if $\UCP(a)=\lambda a$, for some $|\lambda|=1$, then $a\in \Me$. See 
\cite{miza}, Corollary 2.2 for more details.

\begin{theorem}{\rm{(See \cite{wolf}, Theorem 6.7)}}\label{thm:wolf}
Let $\UCP:\Md\rightarrow \Md$ be a unital channel such that it has trivial peripheral spectrum, that is $\rm{Spec}_{\UCP}\cap \mathbb{T}=\{1\}$, then 
$\displaystyle\lim_{n\to\infty}\UCP^{n}(x)=\rm{Tr}(x)\frac{1}{d}$ for all $x\in \Md$.
\end{theorem}
Channels with trivial peripheral spectrum are known as ``primitive" channels. These maps are generic in the sense that they are dense in the set of channels. Following \cite{miza}, Corollary 3.5, a unital channel $\UCP$ is primitive if and only if the stabilized multiplicative domain is trivial, that is, $\mathcal{M}_{\UCP^\infty}=\mathbb{C}1$, where $1$ is the identity matrix in $\Md$.  

With all the necessary background we are ready to write down the main result of this section:
\begin{theorem}\label{thm-PPT-finite-index}
Every unital PPT channel has finite index of separability.
\end{theorem}
\begin{proof}
Consider the channels $\UCP_1,\cdots,\UCP_k$ in Theorem \ref{thm:PPT_Minf_decomp}. Each channel has trivial multiplicative domain, so by \cite{miza}, its peripheral spectrum is trivial. Thus, we can apply Theorem \ref{thm:wolf} and conclude that $\lim_{n}\UCP_i^n(x)=\Tr(x)\tfrac{1}{d_i}$. Then there will be some finite $\ell_i$ such that $\UCP_i^{\ell_i}$ is close enough to $\Tr(x)\tfrac{1}{d_i}$ that their Choi matrices are within a distance of $\frac{1}{d_i(d_i+1)}$. By Theorem \ref{gurvits}, this means that the Choi matrix of $\UCP_i^{\ell_i}$ is separable, that is $\UCP_i^{\ell_i}$ is entanglement breaking. Letting $\ell=\max\{\ell_1,\cdots,\ell_k\}$, then 
\[\UCP^{n\ell}=\UCP_1^{\ell}\oplus\cdots\oplus\UCP_k^{\ell}\]
will be a direct sum of entanglement breaking channels, and thus is entanglement breaking.
\end{proof}

\section{Beyond PPT maps}\label{sec:Beyond PPT maps}
PPT maps turn out to be a special case of channels with finite index of separability. 
This turns out to be, topologically, an abundant class of channels. We need the following result:
\begin{theorem}[see \cite{stormer}, corollary 7.5.5]\label{stormer-cor}
Let $\Phi$ be a positive map such that $\|\Phi\|\leq 1$, then $x\mapsto Tr(x)1+\Phi(x)$ is entanglement breaking.
\end{theorem}

\begin{theorem}\label{thm-denseness}
The set of unital channels on $\Md$ which have finite index of separability is dense (C.B topology) in the set of unital channels. 
\end{theorem}
\begin{proof}
Let $\Phi$ be a unital channel. We will show that given $\delta>0$, there is a unital channel $\UCP$ approximating $\Phi$ within $\delta$ in CB norm and $\UCP$ has finite index of separability.

To this end, let us define the map $\UCP(x)=(1-\frac{\delta}{2})\Phi(x)+\frac{\delta}{2} \Omega(x)$, $\forall x\in \Md$,  where $\Omega(x)=\tr(x)\frac{1}{d}$. It is easily verified that $\UCP$ is a unital channel and  
\[\|\UCP-\Phi\|_{\rm{cb}}\leq \frac{\delta}{2}\|(\Phi+\Omega)\|_{\rm{cb}}\leq\delta. \]
For simplicity call $a=(1-\frac{\delta}{2})$. As $\Phi\circ\Omega=\Omega\circ\Phi=\Omega$, it follows that $\UCP^n=a^n\Phi^n+(1-a^n)\Omega$ for every $n\ge 1$. So
\[\UCP^n=(1-a^n)(\Omega +\frac{a^n}{1-a^n}\Phi^n)=\frac{(1-a^n)}{d}\left(Tr(x)1+\frac{da^n}{1-a^n}\Phi^n\right).\]
As $a^n\to 0$ as $n\to\infty$, it follows that $\|\frac{da^n}{1-a^n}\Phi^n\|\leq 1$ for large $n$ and hence $\Ep^n$ is entanglement breaking for large $n$ by Theorem \ref{stormer-cor}.  
\end{proof}

\subsection{Schur Channels}
Next we provide some concrete examples of maps which are not PPT, yet they have finite index of separability. We start with a definition of Schur channels. Recall for two $d\times d$  matrices $a=(a_{i,j})$ and $b=(b_{i,j})$, the Schur product is defined as $a\circ b=(a_{i,j}b_{i,j})$.
\begin{definition}
Given a matrix $b\in \Md$, we define a map $T_b(x)=b\circ x$ on $\Md$. Such maps are called Schur product maps.
\end{definition}
It is well known that $T_b$ is completely positive if and only if the matrix $b\geq 0$. Moreover, if $b$ has all its diagonal entries equal to 1, then $T_b$ is a unital channel.
\begin{prop}\label{Prop:Schur-characterize}
For a Schur Channel $T_b$ on $\Md$, the following statements are equivalent:
	\begin{enumerate}
		\item
		$T_b$ is PPT.
		\item
		$T_b$ is entanglement breaking.
		\item
		$b=1$, the identity element in $\Md$.
	\end{enumerate}
\end{prop}
\begin{proof}
	$(2\Rightarrow 1)$ is well-known. For $(1\Rightarrow 3)$, a characterization of the multiplicative domain of Schur channels has been put forward in \cite{levick}. It is given below 
	 \[x=(x_{i,j})\in \mathcal{M}_{T_b} \Leftrightarrow x_{i,j}\neq 0, \ \rm{whenever} \ |b_{i,j}|=1 .\] Hence it is clear that the multiplicative domain of $T_b$ contains  the algebra of diagonal matrices of $\Md$. 
 Since the matrix units $E_{ii}\in \mathcal{M}_{T_b}$ for all $i$ and noting that $E_{ij}=E_{ii}E_{ij}E_{jj}$ we get $T_b(E_{ij})=0$ by Corollary  \ref{lemma:coCP_md_struct}. Thus $b=1$ .\par
	To prove $(3\Rightarrow 2)$, if $b=1$, then $T_b$ can be written as $T_b(x)=\sum_{i=1}^dE_{ii}xE_{ii}$. Since these are rank-1 operators, $T_b$ is entanglement breaking by one of the equivalent criteria given in \cite{entng-brkng}. Hence, $\UCP$ is also entanglement breaking.
\end{proof}
The above proposition suggests that there are no non-trivial Schur channels that are PPT, a fact which was proved in \cite{K-M-P} using different method.
\begin{ex}[Non- PPT Maps having finite index of Separability]\label{ex1}
Let
\[b=\begin{pmatrix}1&1&0\\1&1&0\\0&0&1\end{pmatrix},\, u=\begin{pmatrix}0&1&0\\0&0&1\\1&0&0\end{pmatrix},\]
and set $\UCP=\style{U}\circ T_b$, where $\style{U}(x)=uxu^*$, for all $x$. Hence $\UCP$ is a unital channel.
One computes 
\[\Me=\leftset\begin{pmatrix} a_{11}&a_{12}&0\\a_{21}&a_{22}&0\\0&0&a_{33}
\end{pmatrix}\midsetl a_{ij}\in\C\rightset.\]
Clearly it is a non-abelian subalgebra of $\mathcal{M}_3$ and hence $\UCP$ can not be PPT by the Theorem \ref{Thm-mult-abel-PPT}.
But we have that
\[\UCP\begin{pmatrix} x_{11}&x_{12}&x_{13}\\x_{21}&x_{22}&x_{23}\\x_{31}&x_{32}&x_{33}\end{pmatrix}=\begin{pmatrix}x_{33}&0&0\\0&x_{11}&x_{12}\\0&x_{21}&x_{22}\end{pmatrix}.\]
which, by Lemma \ref{lem:md_composition}, gives $\mathcal{M}_{\UCP^2}=D_3$, the algebra of diagonal matrices. Further, we observe that $\UCP^2(\mathcal{M}_3)=D_3$, that is the range of $\UCP^2$ is contained inside the abelian C$^*$-algebra $D_3$, hence by Lemma \ref{abelian-range}, $\UCP^2$ must be entanglement breaking. Thus, $\UCP$ is not PPT, but $\UCP^2$ is entanglement breaking.\par
We can extend this example to higher dimensions: Set $b=\begin{pmatrix} J_{d-1}&0\\0&1\end{pmatrix}$ (where $J_{d-1}$ is the $d-1\times d-1$ matrix of all ones), and set $u$ to be the permutation corresponding to $(123\cdots d)$ in the symmetric group on $\{1,2,\cdots,d\}$. Then setting $\UCP=\style{U}\circ T_b$, we see that $\UCP$ is not PPT, and neither is $\UCP^n$ for $1\leq n\leq d-2$, but $\UCP^{d-1}$ will be entanglement breaking. 
\end{ex}
Note that above example suggests that the converse of the PPT conjecture is false, that is, if $\UCP$ is channel such that $\UCP^2$ is entanglement breaking, then $\UCP$ need not be PPT.

\section{Asymptotically Entanglement Breaking}\label{sec:Asymptotically Entanglement Breaking}
In this section we show the existence of maps that don't have finite index of separability, but asymptotically they are entanglement breaking.  
\begin{definition}We call a linear map $\UCP:\Md\rightarrow\Md$ `asymptotically entanglement breaking' if there is a sequence $\{n_j\}_{j}^{\infty}\subset \mathbb{N}$ such that $\displaystyle\lim_{j\to\infty} \UCP^{n_j}$ is entanglement breaking, where the limit is in the bounded weak (BW) topology.
\end{definition}
In a very recent article \cite{K-M-P}, it was shown that PPT maps are asymptotically entanglement breaking. Indeed, they showed that a limit point of iterates of a PPT map becomes entanglement breaking. This could also be deduced from the work of Lami and Giovannetti \cite{Lami-G}.
In this section we show that there are plenty of channels that are asymptotically entanglement breaking. 
\begin{theorem}\label{asympt-entgl}
Let $\UCP:\Md\rightarrow\Md$ be a unital channel. Then $\UCP$ is asymptotically entanglement breaking map if and only if the stabilized multiplicative domain $\mathcal{M}_{\UCP^\infty}$ is abelian. 
\end{theorem}
\begin{proof}
Note that Kuperberg in \cite{kup} proved that for a unital completely positive map $\UCP$, there is subsequence $n_1,n_2. \cdots,$ such that $\displaystyle\lim_{j\to\infty}\UCP^{n_j}$ converges to a unique conditional expectation $P:\Md\rightarrow\Md$ such that $P$ is completely positive and also 

\[\rm{Range}(P)=\rm{Span}\{a\in \Md: \UCP(a)=\lambda a, \ |\lambda|=1\},\] that is, the range is the span of all the peripheral eigen operators of $\UCP$.  From the Theorem 2.5 in \cite{miza}, we get $\mathcal{M}_{\UCP^\infty}$ is the algebra generated by the set $\rm{Span}\{a\in \Md: \UCP(a)=\lambda a, \ |\lambda|=1\}$. Hence by the hypothesis, $\rm{Range}(P)$ is contained in an abelian C$^*$-algebra. Hence by the Lemma \ref{abelian-range}, we have $P$ is entanglement breaking. Thus $\UCP$ is asymptotically entanglement breaking.

Conversely, let there exist a subsequence of $m_1<m_2<\cdots$ such that \[\lim_{k\to\infty} \UCP^{m_k}=Q,\] then it follows that the idempotent $Q$ and $\UCP^n$ commutes for every $n\geq 1$. It is evident that \[\lim_{n}\|\UCP^n-\UCP^n\circ Q\|=0.\]
Now choosing particularly the subsequence $n_1<n_2<\cdots$ for which we get the conditional expectation $P$ onto the subalgebra $\M_{\UCP^\infty}$, we get 
\[\lim_{j}\|\UCP^{n_j}-\UCP^{n_j}\circ Q\|=0.\]
Now passing to any subsequence of the sequence $n_1<n_2<\cdots$ we may conclude that 
\[\lim_{i}\|P-\UCP^{n_i}\circ Q\|=0.\]

As $Q$ is entanglement breaking and composing it with any 
completely positive map yields another entanglement breaking map, it is immediate that $P$ is entanglement breaking since the set of entanglement breaking channel is a closed set. Hence by Lemma \ref{abelian-range} we get $\M_{\UCP^\infty}$ is abelian.
 
\end{proof}

\begin{remark}
If any of the limit points of the set $\displaystyle\{\UCP^n\}_{n=1}^{\infty}$ are entanglement breaking, then the above proof technique can be used to show that any other limit point of the set $\{\UCP^n\}$ will be a limit of entanglement breaking maps, which is again entanglement breaking.
This fact was first proved in \cite{Lami-G}, Proposition 20. Note that Lami-Giovannetti in \cite{Lami-G} proved Theorem 24 which is essentially equivalent to Theorem \ref{asympt-entgl}. Indeed, they prove that a channel $\UCP$ is asymptotically entanglement saving if and only if the stabilized multiplicative domain (of $\UCP^*$) is non-abelian.    
\end{remark}

We next demonstrate a large class of maps that have the property mentioned in the above theorem.
\begin{definition}
A positive linear map $\UCP$ is called irreducible if $\UCP(p)\leq \lambda p$ holds for a projection $p$ and $\lambda>0$, then $p\in \{0,1\}$, that is, $p$ must be a trivial projection.
\end{definition}
Irreducible maps are generic in the sense that these maps are dense in the set of all positive linear maps acting on $\Md$.
\begin{cor}\label{cor:irreducible}
Every unital irreducible channel $\UCP:\Md\rightarrow\Md$  is asymptotically  entanglement breaking.
\end{cor}
\begin{proof}
Note that by the aid of Perron-Frobenius theory of irreducible positive maps, we know that (see \cite{evans-krohn}) the peripheral eigen operators of an irreducible channel are generated by a single unitary. It is proved in the Lemma 3.4 in \cite{miza} that for a unital irreducible channel $\UCP$, we have 
$\mathcal{M}_{\UCP^\infty}$ is an abelian C$^*$-algebra. Hence by the Theorem \ref{asympt-entgl} we have the result.
\end{proof}
We end this section by providing an example of a channel $\UCP$ which is asymptotically entanglement breaking but does not have finite index of separability, ensuring that $\mathcal{M}_{\UCP^\infty}$ being abelian is not sufficient for $\UCP$ to have finite index of separability.
\begin{ex}\label{ex2}
Let $\UCP:\mathcal{M}_2\rightarrow\mathcal{M}_2$ be defined as a Schur product channel $\UCP(x)=b\circ x$, where $b=\begin{pmatrix}
1 & \lambda\\
\lambda & 1\\
\end{pmatrix}$, with $0<\lambda<1$.
It follows that the stabilized multiplicative domain  $\mathcal{M}_{\UCP^\infty}=\Big\{\begin{pmatrix}
c_1 & 0\\
0  & c_2\\
\end{pmatrix}| c_1,c_2\in \mathbb{C}\Big\}$, which is clearly abelian. However following Proposition \ref{Prop:Schur-characterize}, it is evident that $\UCP^n$ is not entanglement breaking for any $n\geq 1$ as $\UCP^n(x)=b_n\circ x$, where $b_n=\begin{pmatrix}
1 & \lambda^n\\
\lambda^n & 1\\
\end{pmatrix}$.   

\end{ex} 
\begin{remark} 
Note that since the irreducible channels are dense in the set of all quantum channels, Corollary \ref{cor:irreducible} ensures that the set of unital channels that are asymptotically entanglement breaking is also a dense subset of the unital channels. This result, combined with Theorem \ref{thm-denseness}, demonstrates the richness of the class of eventually entanglement breaking maps.
\end{remark}

\section{Discussion}\label{sec:PPT2}
The requirement of unitality of PPT channels can be relaxed in some cases if some properties of the adjoint map are exploited. Note that our method guarantees the existence of a number $n$ for a unital PPT channel $\UCP$ on $\Md$ such that $\UCP^n$ s entanglement breaking. However, a uniform bound could not be found. A uniform upper bound for the multiplicative index may provide an upper bound for this number. In a recent article \cite{Sam-Miza}, Theorem 3.8, such a bound for multiplicative index of a channel was proposed. 

Also analyzing the structure of PPT maps in Section \ref{sec:Structure of PPT Maps}, it can be realized that to prove the PPT-squared Conjecture (\ref{conj-PPT}) it is enough to prove the same for unital and trace preserving PPT maps which have trivial multiplicative domain.  
\section{Acknowledgement}
We would like to thank the referee for pointing out a simplified proof of Theorem \ref{thm-denseness}.  
M.R is supported by a Post Doctoral Fellowship of the department of Pure Mathematics, University of Waterloo and Institute of Quantum Computing, S.J is supported by a Graduate Research Fellowship of the department of Combinatorics and Optimization and Institute of Quantum Computing and V.P is supported by NSERC Grant Number 03784.  
\bibliography{EEB1}
\bibliographystyle{amsplain}

\end{document}